\documentclass[11pt]{article}

\usepackage[T1]{fontenc}              
\usepackage[latin9]{inputenc}         

\usepackage[a4paper]{geometry}        

\usepackage{url}
\usepackage{hyperref}

\usepackage{booktabs}

\usepackage{graphicx} 
\usepackage{amssymb,amsfonts,amsmath}
\usepackage{amsthm}

\usepackage{natbib}
\usepackage{mathpazo}
\usepackage{microtype}

\usepackage{color}

\usepackage{todonotes}
 
\providecommand\bt{\boldsymbol t}
\providecommand\bq{\boldsymbol q}
\providecommand\bp{\boldsymbol p}

\newcommand\bmu{\boldsymbol \mu}

\newcommand\bSigma{\boldsymbol \varSigma}

\providecommand\bfeta{\boldsymbol \eta}

\providecommand\geomdist{\operatorname{Geom}}
\providecommand\bindist{\operatorname{Bin}}
\providecommand\negbindist{\operatorname{NB}}
\providecommand\normdist{N}
\providecommand\gamdist{\operatorname{Gam}}
\providecommand\poisdist{\operatorname{Pois}}
\providecommand\nakdist{\operatorname{Nak}}
\providecommand\berdist{\operatorname{Ber}}
\providecommand\expdist{\operatorname{Exp}}
\providecommand\catdist{\operatorname{Cat}}
\providecommand\invgamdist{\operatorname{IG}}
\providecommand\sqrtgamdist{\operatorname{SRG}}
\providecommand\gengamdist{\operatorname{GG}}

\providecommand\halfnormdist{\operatorname{HN}}
\providecommand\rayleighdist{\operatorname{Rayl}}
\providecommand\maxboltzdist{\operatorname{MB}}

\providecommand\prob{\operatorname{Pr}}

\newcommand\expect{\operatorname{E}}

\newcommand\var{\operatorname{Var}}

\newcommand\ikl{D_{\text{KL}}}

\newtheorem{prop}{Proposition}

\newcommand{\eqcite}[1]{Equation~{\bf\ref{#1}}}
\newcommand{\ineqcite}[1]{Inequality~{\bf\ref{#1}}}
\newcommand{\tabcite}[1]{Table~{\bf\ref{#1}}}
\newcommand{\seccite}[1]{Section~{\bf\ref{#1}}}

\definecolor{revision}{rgb}{1,0,0} 

\begin{document}

\date{24 April 2026; revised 3 December 2026}

\title{
From Physics to Statistics:
A Simple Route to Exponential Families via Maximum Entropy
}

\newcommand\myaddress[2][]{\relax
    {\noindent{\small$^{#1}$#2}}}

\author{Korbinian Strimmer$^{1}$ 
\thanks{To whom correspondence should be addressed. Email: {\tt korbinian.strimmer@manchester.ac.uk  } 
}
}

\maketitle

\myaddress[1]{Department of Mathematics, University of Manchester, 
Alan Turing Building, Oxford Road, Manchester M13 9PL, UK.}

\newpage

\begin{abstract} 

Exponential families form the backbone of modern statistics and machine learning, but textbooks seldom explain them from first principles in an accessible way. Although minimal sufficiency and maximum entropy, a principle originating in physics, provide core motivation, they are often presented as technical and requiring advanced prerequisites.

Here, a short, self-contained derivation of exponential families based on maximum entropy is presented that is straightforward to carry out, requires only a modest background in information entropy, and avoids technicalities like constrained optimisation.  Two propositions are demonstrated in this fashion:
i)~exponential families with a general base function maximise information entropy with respect to a corresponding base distribution subject to fixed expectations of canonical statistics, and
ii)~exponential families with a constant base function maximise standard information entropy under the same constraints.

Maximum entropy therefore provides a principled foundation for exponential families with minimal prerequisites, highlighting the value of teaching entropy in statistics courses.

\vspace{1cm}


\end{abstract}

\newpage

\section{Introduction}

Exponential families are a widely used class of probability models underpinning much of modern data analysis and machine learning \citep{EfronExpFam2022,Murphy2023}. They form the basis of generalised linear and other regression models and also underlie probabilistic graphical models like Bayesian networks, Ising models, and Gaussian Markov random fields \citep{WainwrightJordan2008}. Exponential families offer many practical and theoretical benefits. For example, they permit effective algorithms for statistical learning, they enable data reduction while preserving information about parameters, and results proven under specific assumptions, such as normality, can often be generalised, at least approximately, to exponential families.

\subsection{Principles underlying exponential families}

The first concepts related to exponential families emerged between 1870 and 1900 in physics during the development of statistical mechanics.
In particular, Boltzmann and Gibbs introduced three key innovations in statistical thermodynamics that also proved to be fundamental to information theory and statistics:
\renewcommand{\theenumi}{\roman{enumi}}
\begin{enumerate}
\item  Distribution-based definitions of entropy, with both uniform and general prior measures, corresponding to standard information entropy and negative Kullback-Leibler (KL) divergence, respectively \citep{Akaike1985},  
\item the general principle of entropy maximisation, to find equilibrium states and least-informative distributions \citep{LevineTribus1979}, and
\item exponential families, with both constant and general base functions, derived as maximum entropy distributions subject to expectation constraints \citep{Jaynes1957a,Jaynes1968}.
\end{enumerate}

In later formal developments of exponential families in statistics, with the canonical form arising from exponential tilting \citep{Esscher1932}, they were motivated by their data-summarisation property as the unique parametric distributions  allowing fixed-dimensional minimal sufficient statistics \citep{Fisher1934}.
This was demonstrated by the Pitman-Koopman-Darmois (PKD) theorems 
\citep{Pitman1936,Koopman1936,Darmois1935}, see also \citet{BarankinMaitra1963}.
The concept of minimal sufficiency was also anticipated in statistical mechanics \citep{Mandelbrot1962}. 

Therefore, both minimal sufficiency and maximum entropy provide two distinct foundational principles for characterising exponential families.

\subsection{Teaching exponential families}

Most modern statistics and machine-learning texts introduce exponential families by stating their distributional form, but few explain the underlying motivating principles.  A few notable exceptions are listed in \tabcite{tab-textbooks}.
 
Among these works,  only \citet{Murphy2023} demonstrates how to use maximum entropy to derive exponential families with a general base function. Except for \citet{Sundberg2019} textbooks at most mention PKD theorems and minimal sufficiency,
but without deriving exponential families along those principles.
In machine learning, entropy maximisation and its connection with exponential families are recognised 
as important, but any derivation is typically limited to exponential families with a constant base function.

The absence of exponential-family derivations in statistics textbooks likely reflects several factors, notably the perceived difficulty of those derivations and the presumed need for advanced prerequisites.  For example, PKD-type theorems can be technically demanding, whereas maximum-entropy arguments assume familiarity with information-theoretic concepts. Furthermore, maximum-entropy methods often rely on constrained optimisation using Lagrange multipliers and functional calculus, topics typically taught in physics or applied mathematics but not core statistics.

\subsection{A simple maximum-entropy route to exponential families} 

This note presents a simple, back-of-the-envelope derivation of exponential families based on maximum entropy. It avoids explicit optimisation, uses only elementary calculations and assumes only a basic background in information entropy.  Moreover, it covers both general exponential families as well as those with a constant base function.

In the following, \seccite{sec-expfam} revisits exponential families and
\seccite{sec-entropy} introduces standard information entropy
and information entropy with respect to a prior measure.
Subsequently, \seccite{intentmaxent} discusses interpretations of entropy as well
as the principle of maximum entropy.
Finally, \seccite{sec-efmaxent} shows how both the general exponential family  and as well as the subset of exponential families with constant base function naturally arise from the principle of maximum entropy.

\begin{table}[t]
\caption{Derivation of exponential families in contemporary textbooks.}
{\centering
\makebox[\textwidth]{
\begin{tabular}{l r ll }
\toprule
Reference         & Field & via PKD theorems & via maximum entropy \\
\midrule
\citet{Murphy2023} & ML  & Sections 2.4.5 and 5.3.4 (mention) & Section 2.4.7 (EF) \\
\citet{EfronExpFam2022} & ST & Section 1.3 (mention)  & --- \\
\citet{McElreath2020}   & ST & --- & Section 10.2.1 (mention) \\
\citet{Sundberg2019}  & ST & Proposition 3.4 & --- \\
\citet{Little2019} & ML  & Section 4.5 (mention) & Section 4.5 (EFC) \\
\citet{Golan2018} &IT & --- & Chapter 4 (EFC)\\
\citet{Murphy2012} & ML  & Section 9.2.4 (mention) & Section 9.2.6 (EFC) \\
\cite{CoverThomas2006} & IT & --- & Theorem 12.1.1 (EFC) \\
\bottomrule
\end{tabular}}}
\\{\it Abbreviations:} ST statistics, ML machine learning, IT information theory,
EF exponential family with general base function, EFC exponential family with constant base function, PKD Pitman-Koopman-Darmois.
\label{tab-textbooks}\\
\end{table}

\section{Exponential families in a nutshell}
\label{sec-expfam}

Many commonly used distributions in statistics are exponential families.  This section provides
a summary of some of their key properties. For modern, comprehensive treatments see, for example, \citet{Sundberg2019} and \citet{EfronExpFam2022}.
\tabcite{tab-expfam} shows some examples of univariate exponential families.

\subsection{Distributional form}

\begin{table}[pt!]
\caption{Examples of univariate exponential families.}
{\centering
\makebox[\textwidth]{
\begin{tabular}{l l l l l }
\toprule
Distribution  & $h(x)$ & $z(\bfeta)$ & $\bfeta$ & $\bt(x)$   \\
  & $\text{supp}(h)$ & & & $\bmu_{\bt}=\expect(\bt(x))$  \\

\midrule
$\bindist(n, \theta)$ & $\binom{n}{x}$ & $(1+\exp \eta)^n$ & $\log\left(\frac{\theta}{1-\theta}\right)$ & $x$  $\phantom{\begin{pmatrix} x \\ x^2\end{pmatrix}}$  \\
& $\{0, 1, \dots, n\}$ & $\phantom{l}=(1-\theta)^{-n}$ &  &  $n \theta$  \\

\midrule
$\negbindist(n, \lambda)$ & $\frac{\Gamma(x+n)}{x! \, \Gamma(n)}$ & $(1-\exp \eta)^{-n}$ & $\log \lambda$ & $x$  $\phantom{\begin{pmatrix} x \\ x^2\end{pmatrix}}$  \\
& $\{0, 1, \dots\}$ & $\phantom{l}=(1-\lambda)^{-n}$ &  &  $n \lambda / (1-\lambda)$  \\

\midrule
$\poisdist(\mu)$  & $\frac{1}{x!}$ & $\exp(\exp \eta) $ & $\log \mu$ & $x$ $\phantom{\begin{pmatrix} x \\ x^2\end{pmatrix}}$   \\
  & $\{0, 1, \ldots\}$  & $\phantom{l}= \exp \mu $ & &$\mu$ \\

\midrule
$\normdist(\mu,\sigma^2)$  & $1$ & $(-\frac{\pi}{\eta_2})^{1/2}$ $\exp\left(-\frac{\eta_1^2 }{4 \eta_2}\right)$ & $\begin{pmatrix}  \mu/\sigma^{2} \\ -1/(2\sigma^{2})\end{pmatrix}$ & $\begin{pmatrix} x \\ x^2\end{pmatrix}$ \\
 & $\mathbb{R}$ & $\phantom{l}= (2 \pi \sigma^2 )^{1/2}$ $\exp\left( \frac{\mu^2}{2 \sigma^2} \right)$ &  & $\begin{pmatrix} \mu \\ \sigma^2 + \mu^2 \end{pmatrix}$   \\

\midrule
$\gamdist(\alpha, \theta)$  & $1$&$(-\eta_1)^{-\eta_2-1}$ $\Gamma(\eta_2+1)$  & $\begin{pmatrix} -1/\theta \\ \alpha-1 \end{pmatrix}$ & $\begin{pmatrix}  x \\ \log x \end{pmatrix}$  \\
 & $\mathbb{R}^{+}$ &$\phantom{l}=\theta^{\alpha} \, \Gamma(\alpha)$ &  &$\begin{pmatrix} \alpha \theta \\ \psi^{(0)}(\alpha) +\log\theta\end{pmatrix}$   \\

\midrule
$\gengamdist(\alpha, \beta, \lambda)$  & $1$&$ \frac{1}{|\lambda|} (-\eta_1)^{-(\eta_2+1)/\lambda} \Gamma\left(  \frac{\eta_2+1}{\lambda}   \right)$  & $\begin{pmatrix} -1/\beta^{\lambda} \\ \lambda \alpha -1 \end{pmatrix}$    & $\begin{pmatrix}  x^\lambda \\ \log x \end{pmatrix}$  \\
  & $\mathbb{R}^{+}$ &$\phantom{l} = \frac{1}{|\lambda|} \beta^{\lambda \alpha}\, \Gamma(\alpha)$  & & $\begin{pmatrix} \alpha \beta^{\lambda} \\\frac{1}{\lambda} \psi^{(0)}(\alpha) +\log \beta \end{pmatrix}$\\

\bottomrule
\end{tabular}}
}
{
\\ {\it Distributions listed above:}
Binomial $\bindist(n, \theta)$,
negative binomial $\negbindist(n, \lambda)$,
Poisson  $\poisdist(\lambda)$,
normal  $\normdist(\mu,\sigma^2)$,
gamma  $\gamdist(\alpha, \theta)$,
generalised gamma $\gengamdist(\alpha, \theta, \lambda)$.\\
{\it  Related distributions (special cases):}
Bernoulli  $\berdist(\theta) = \bindist(n=1, \theta)$,
geometric $\geomdist(\lambda) = \negbindist(n=1, \lambda)$,
gamma $\gamdist(\alpha, \theta) = \gengamdist(\alpha, \beta=\theta, \lambda=1)$,
scaled chi-squared  $s^2 \,\chi^2_k = \gamdist(\alpha=k/2, \theta=2 s^2)$,
exponential  $\expdist(\theta) = \gamdist(\alpha=1, \theta)$,
inverse gamma $\invgamdist(\alpha, \beta) = \gengamdist(\alpha, \beta, \lambda=-1)$, 
square-root gamma  $\sqrtgamdist(\alpha, \beta) = \gengamdist(\alpha, \beta, \lambda=2)$, 
Nakagami  $\nakdist(\alpha, \Omega) = \sqrtgamdist(\alpha, \beta=\sqrt{\Omega/\alpha})$,
scaled chi  $s \,\chi_k = \sqrtgamdist(\alpha=k/2, \beta=\sqrt{2 s^2})$,
half-normal  $\halfnormdist(s) = s \,\chi_1$,
Rayleigh  $\rayleighdist(s) = s \,\chi_2$,
Maxwell-Boltzmann  $\maxboltzdist(s)= s \,\chi_3$.\\
{\it  Special function:}
$\psi^{(0)}(x) =\frac{d}{dx} \log \Gamma(x)$ denotes the digamma function.
}
\label{tab-expfam}\\
\end{table}

A family of distributions $P(\bfeta)$ for a random variable $x$ is an exponential family (EF) if it is generated by  exponential tilting of a base distribution $B=P(\bfeta_0)$, producing a probability density 
or mass function (pdmf) of the form
\begin{equation}
\begin{split}
p(x|\bfeta) &=   h(x)\, e^{\langle \bfeta, \bt(x) \rangle}/ z(\bfeta)\\
       & =  h(x)\, e^{\langle \bfeta, \bt(x) \rangle -a(\bfeta)}\\
\end{split}
\end{equation}
where 
\begin{itemize}
\item $\bt(x)$ are the canonical statistics, of fixed and typically low dimension,
\item $\bfeta$ are the corresponding canonical parameters,
\item $h(x)$ is a positive base function, also encoding the support of the pdmf, 
\item $e^{\langle \bfeta, \bt(x) \rangle}$ is the exponential tilt,
\item $z(\bfeta)$ is the partition function (the normaliser) and
$a(\bfeta) = \log z(\bfeta)$ the corresponding log-partition function (the log-normaliser).
\end{itemize}
The inner product notation $\langle \cdot,\cdot \rangle$ in the  tilting factor allows 
canonical statistics and parameters to be scalars, vectors or matrices.

The base function $h(x)$ is usually unnormalised and may not even have a finite normalising constant.
The base distribution $B=P(\bfeta_0)$ may be chosen as any member of the exponential family, with
$\bfeta_0$ denoting its parameter value. The corresponding
base pdmf is 
\begin{equation}
\begin{split}
b(x) = p(x|\bfeta=\bfeta_0) & = h(x)\, e^{\langle \bfeta_0, \bt(x) \rangle } /z(\bfeta_0) \\
                            & = h(x)\, e^{\langle \bfeta_0, \bt(x) \rangle -a(\bfeta_0)}\\ \,.
\end{split}
\end{equation}

The pdmf of the exponential family $P(\bfeta)$ can then be expressed in terms of $b(x)$, rather than $h(x)$, as
\begin{equation}
\begin{split}
p(x|\bfeta) &=   b(x)  e^{\langle \bfeta-\bfeta_{0}, \bt(x) \rangle}
                / \, \left( \frac{z(\bfeta) }{ z(\bfeta_{0})} \right)\\
       & =  b(x) \, e^{ \langle \bfeta - \bfeta_{0},  \bt(x)\rangle - a(\bfeta) +a(\bfeta_{0})} \\
\end{split}
\label{eq:expfampdfb}
\end{equation}
A natural choice for the base distribution is to set $\bfeta_0 = 0$, which is
a permissible parameter value whenever $z(0)$ is finite, yielding as base pdmf $b(x)=h(x)/z(0)$.

A subset of exponential families (denoted here as EFC) can be written with a constant base
function $h(x)=1$, see \tabcite{tab-expfam} for examples. This results in the simplified
pdmf 
\begin{equation}
\begin{split}
p(x|\bfeta) & =   e^{\langle \bfeta, \bt(x) \rangle}/ z(\bfeta)\\
            &= e^{\langle \bfeta, \bt(x) \rangle -a(\bfeta) }\,.
\end{split}
\label{eq:expfampdfc}
\end{equation}

If the elements in $\bt(x)$ are affinely independent the representation of
the exponential family is minimal, otherwise the representation is nonminimal.
Natural exponential families \citep{MorrisLock2009} are obtained by tilting by
identity with $t(x) = x$.

In physics, exponential families are familiar under the term canonical distribution or Boltzmann distribution\footnote{The term Boltzmann distribution as a synonym for exponential families
should not be confused with the more specific Maxwell-Boltzmann distribution $\maxboltzdist(s)$. The latter is
used for modelling molecular speeds and is a particular instance of
the square-root gamma distribution $\sqrtgamdist(\alpha, \beta)$, also known as 
scaled chi distribution $s \,\chi_k$, and is itself an exponential family (\tabcite{tab-expfam}).}, with
the exponential tilt known as Boltzmann factor and the base function $h(x)$
as degeneracy factor \citep{Reif1965,Mandl1988}.

\subsection{Partition function}

The partition function  
\begin{equation}
z(\bfeta) = \int_x \, e^{ \langle \bfeta, \bt(x) \rangle}\, h(x) \, dx
\end{equation}
ensures that
 $p(x|\bfeta)$ is normalised (for discrete $x$ replace the integral by a sum).

The set of
values of $\bfeta$ for which the partition function is finite, $z(\bfeta)  < \infty$, and
hence for which $p(x|\bfeta)$
is well defined, comprises the parameter space of the
exponential family. Some choices of $h(x)$ and $\bt(x)$ do not yield a finite normalising factor for any $\bfeta$ and hence these cannot be used to form an exponential family.  

The log-partition function $a(\bfeta)= \log z(\bfeta)$ allows to compute the cumulants of the canonical statistics $\bt(x)$ expressed in terms of the canonical statistics $\bfeta$.
In particular, its gradient
yields the mean
\begin{equation}
\expect( \bt(x) )  = \bmu_{\bt}(\bfeta)  = \nabla a(\bfeta)
\end{equation}
and the Hessian matrix the variance
\begin{equation}
\var( \bt(x) )  = \bSigma_{\bt}(\bfeta)  = \nabla \nabla^T a(\bfeta)\,.
\end{equation}

The mean $\bmu_{\bt}$ of the canonical
statistics $\bt(x)$ is often used as alternative parametrisation instead of
the canonical parameters $\bfeta$. Indeed, in many cases some of the expectation parameters $\bmu_{\bt}$  correspond to conventional parameters, such as the mean $\expect(x)$ for
$t(x)=x$. For an exponential family with minimal representation there is a one-to-one map between the canonical parameters $\bfeta$ and the expectation parameters\  $\bmu_{\bt}$. 

The term partition function and the technique of differentiating its logarithm
also originated in statistical mechanics \citep{Reif1965,Mandl1988}.

\subsection{Exponential families and change of variables}
\label{sec:efvarchange}

Assume that a random variable $x$ follows an exponential-family distribution
and that $y(x)$ is an invertible transformation.
The resulting pdmf for the transformed variable $y$ is given by
\begin{equation}
p(y|\bfeta) = h_y(y)\, e^{ \langle \bfeta, \bt_y(y) \rangle } / z(\bfeta)
\end{equation}
with transformed canonical statistics
\begin{equation}
\bt_y(y) = \bt_x(x(y))\,.
\end{equation}
For a discrete random variable the base function changes to
\begin{equation}
h_y(y) = h_x(x(y))
\end{equation}
whereas for a continuous random variable the base function becomes
\begin{equation}
h_y(y) = \left| D x(y) \right|\,h_x(x(y))
\end{equation}
where $D x(y)$ is the Jacobian matrix of the inverse transformation  $x(y)$.

Thus, for both discrete and continuous random variables the exponential family is closed under
under a change of variables, and the exponential-family form is preserved.   
However, since a change of variables transforms the base function, 
 exponential families with constant base function (EFC) will normally be transformed to general exponential family (EF) form.

In some cases, depending on the canonical statistics and the transformation, the  Jacobian factor in the modified base function can be integrated back into the tilting factor, so that an EFC remains an EFC even after a change of variables.
An example is the relationship between the gamma distribution $\gamdist(\alpha, \theta)$ and 
the generalised gamma distribution $\gengamdist(\alpha, \beta, \lambda)$ (\tabcite{tab-expfam}).
If $x \sim \gamdist(\alpha, \theta)$  then 
$y=x^{1/\lambda} \sim \gengamdist(\alpha, \beta=\theta^{1/\lambda}, \lambda)$.
The gamma distribution is an EFC and the generalised gamma distribution is also an EFC
because the log-Jacobian factor
$\log \left| D x(y) \right| = \log \left| \lambda \right| + (\lambda-1) \log y$
resulting from the power transformation
is linear in the canonical statistic $\log y$.

\section{Information entropy}
\label{sec-entropy}

Entropy, a fundamental concept originating in physics, plays a central role in
information theory and statistical learning. This section introduces information
entropy through proper scoring rules, particularly log-loss, and combinatorial arguments.
It presents entropy with both uniform and general prior measures and discusses their properties,
in particular concavity and behaviour under change of variables.

\subsection{Definition of information entropy}
\label{sec:defent}

The information entropy $H(Q)$ of the distribution $Q$ with pdmf $q(x)$ is given by
\begin{equation}
H(Q)= - \expect_Q\left(\log q(x)\right) \,. 
\label{eq:infoent}
\end{equation}
Information entropy is also called Shannon entropy (for discrete
distributions) and differential entropy (for continuous distributions).
In physics it is commonly known as Boltzmann-Gibbs entropy.

Information entropy is closely linked to the logarithmic scoring rule or log-loss
\begin{equation}
S(x, P) = - \log p(x)
\end{equation}
that assesses the model $P$ mit pdmf $p(x)$ by assigning a numerical score based on $P$
and an observed outcome $x$. Using $S(x, P)$ the information entropy is
\begin{equation}
H(Q)= \expect_{Q}(S(x, Q)) \,. 
\label{eq:scoreent}
\end{equation}

The mean log-loss
\begin{equation}
R(Q,P) =  - \expect_{Q} (\log p(x) )
\label{eq:meanlogloss}
\end{equation}
denotes the risk of $P$ under $Q$. Using $S(x, P)$ it can be 
written as
\begin{equation}
R(Q,P) = \expect_{Q}(S(x, P))\,.
\label{eq:risk}
\end{equation}
Note that $R(Q, Q) = H(Q)$.

The mean log-loss is uniquely minimised
for $P=Q$, with the minimum risk equal to the entropy $H(Q)$,
so that
\begin{equation}
R(Q, P) \geq H(Q) \,,
\label{eq:properness}
\end{equation}
which is referred to as the properness
inequality\footnote{For general scoring rule $S(x, P)$, 
\ineqcite{eq:properness}  establishes that $S(x, P)$ is proper. Each proper
scoring rule $S(x, P)$ has an associated scoring rule entropy $H(Q)$, defined by \eqcite{eq:scoreent}, 
acting as generalised information entropy \citep{Dawid2007,DawidMusio2014}.}.
For the log-loss specfically, it is also known as Gibbs' inequality
and follows from Jensen's inequality, see 
Theorem~2.6.3 in \citet{CoverThomas2006}.

Loss functions, and thus scoring rules, are only  defined up to an
equivalence class
\begin{equation}
S(x, Q)^{\text{equiv}} = k S(x, Q) + c(x)\,,
\end{equation}
where $k >0$ is positive scaling factor.
Hence, information entropy and the mean log-loss are also only defined
up to an affine transformation, 
\begin{equation}
H(Q)^{\text{equiv}}= k H(Q) + c
\end{equation}
with $c = \expect_Q(c(x))$ and
\begin{equation}
R(Q, P)^{\text{equiv}}= k R(Q, P) + c\,.
\end{equation}
The factor $k$ determines the units of entropy as well as
the base of the logarithm. Physical entropy uses the Boltzmann constant 
$k_B=1.380649 \times 10^{-23} J/K$ as scaling factor, giving physical
entropy units of energy/temperature, whereas information entropy selects
$k$ to set the units of entropy to bits (binary logarithm) or nats
(natural logarithm). The additive constant $c$ 
can be chosen to fix a reference point, such as setting the maximum
or minimum entropy to zero.

\subsection{Information entropy with prior measure}

Information entropy with prior measure of the distribution $Q$ with reference
to the distribution $P$ is given by 
\begin{equation}
\begin{split}
G(Q, P) &=  H(Q) - R(Q, P) \\
 &= -\expect_Q\log\left(\frac{q(x)}{p(x)}\right)\,.
\end{split}
\end{equation}
By construction, $G(Q, P)$ is nonpositive with a maximum at zero. 

When the reference distribution $P$ is the uniform distribution $U$, the mean log-loss
$R(Q, U)$ is a constant that depends only on $U$, but not on $Q$.
Thus, $G(Q, U) = H(Q)+c$, so information entropy with uniform prior measure
is equivalent to standard information entropy $H(Q)$.

In physics, $G(Q,P)$  is referred to  
as relative Boltzmann-Gibbs entropy \citep[Eq.~14 in][]{PachterYangDill2024}
or generalised Boltzmann-Gibbs-Shannon entropy \citep[p.~226 in][]{Wehrl1978}.
Jaynes introduced
information entropy with prior measure $G(Q,P)$ to statistics \citep{Jaynes1968,Rosenkrantz1983}.

The negative of entropy with prior measure $G(Q, P)$ is the
Kullback-Leibler (KL) divergence \citep{KullbackLeibler1951}
\begin{equation}
\ikl(Q, P) = R(Q, P) - H(Q) = - G(Q, P)\,,
\end{equation}
which in turn is the canonical divergence\footnote{
In general, the divergence associated with a proper scoring rule $S(x, P)$ is a Bregman divergence \citep{Ovcharov2018}.
} induced by the log-loss \citep{Dawid2007}.

The equivalence class for entropy with prior measure 
\begin{equation}
G(Q, P)^{\text{equiv}} =  H(Q)^{\text{equiv}} - H(Q, P)^{\text{equiv}}   =  k G(Q, P)
\end{equation}
only depends on the scale
factor $k$ as the additive constant $c$ in $H(Q)^{\text{equiv}}$ and $H(Q, P)^{\text{equiv}}$ cancels.

\subsection{Combinatorial derivation}
\label{sec:combient}

Entropy can also be obtained by combinatorial considerations.
This was the route taken by Boltzmann in the early days
of statistical mechanics. One of his central insights was the distinction
between the macrostates of a system and their underlying microstates, with entropy as
a logarithmic measure of the  size or volume of a macrostate.
Taking this volume to be proportional to the number of  microstates compatible with the observed macrostate
yields standard information entropy.  A natural generalisation is to measure the volume of a macrostate 
by its probability, which leads to information entropy with prior measure.

This is best illustrated via the multinomial model.
Let $D=\{n_1, \ldots, n_K\}$ be counts for $K$ classes with $n = \sum_{k=1}^K n_k$
and  $\widehat{Q} = \catdist(\hat{\bq})$ be the corresponding empirical categorical
distribution with class frequencies $\hat{q}_k = n_k/n$.
A macrostate corresponds to the counts $D$, or equivalently to the distribution 
$\widehat{Q}$, and each particular allocation of the $n$ elements to the $K$ groups constitutes a microstate.
The number of allocations of $n$ distinct items to $K$ groups compatible with the given counts
$D$ is given by the multinomial coefficient 
\begin{equation}
W_K = \binom{n}{n_1, \ldots, n_K} 
    = \frac {n!}{n_1! \, n_2! \, \ldots \, n_K! } \,,
\end{equation}
so $W_K$ is the multiplicity of the macrostate~$D$. 
The probability to observe $D$ assuming prior class probabilities  $\bp=(p_1, \ldots, p_K)^T$
is 
\begin{equation}
\prob(D| \,\bp) =  W_K \times \prod_{k=1}^K  p_k^{n_k}  \,.  
\end{equation}

Taking the logarithm of the multiplicity $W_K$
directly yields, with identity for large $n$, standard information entropy
\begin{equation}
\frac{1}{n} \log W_K \approx  H(\widehat{Q}) = -\sum_{k=1}^K \hat{q}_k  \log \hat{q}_k\,.
\end{equation}
Likewise, noting that
\begin{equation}
\frac{1}{n}\log\left(\prod_{k=1}^K  p_k^{n_k}\right) = \sum_{k=1}^K\hat{q}_k  \log p_k = -R(\widehat{Q}, P)
\end{equation}
with $P = \catdist(\bp)$,
taking the logarithm of the probability $\prob(D| \,\bp)$
yields  information entropy 
\begin{equation}
\frac{1}{n} \log  \prob(D| \,\bp)  \approx G(\widehat{Q}, P)= -\sum_{k=1}^K \hat{q}_k  \log \left( \frac{\hat{q}_k}{p_k}\right)
\end{equation}
with prior measure $P$.

The combinatorial derivations of the formulas for information entropy $H(Q)$ and  information entropy with prior measure $G(Q,P)$ are due to Boltzmann in 1877 and 1878, respectively. \citet{Akaike1985} provides a historical account of this innovation in physics from a statistical perspective.  These results also form the foundation
of the theory of large deviations, see \citet{Ellis1999} for an overview.

\subsection{Concavity of entropy}

A key property of information entropy $H(Q)$ is that it is strictly concave in~$Q$. This means that
\begin{equation}
H( Q_{\lambda} ) > (1-\lambda) H(Q_0) + \lambda H(Q_1)
\label{eq:entconc}
\end{equation}
for the mixture $Q_{\lambda} = (1-\lambda) Q_0 + \lambda Q_1$
with $0 < \lambda < 1$ and $Q_0 \neq Q_1$. This  follows from
the strict concavity of the function $-x \log(x)$.

Likewise, information entropy with prior measure $G(Q, P)$ is also strictly concave in~$Q$.
Therefore,
\begin{equation}
G( Q_{\lambda}, P ) > (1-\lambda) G(Q_0, P) + \lambda G(Q_1, P)
\label{eq:entpconc}
\end{equation}
The concavity of $G(Q,P)$ in $Q$ derives from the concavity of $H(Q)$
and that the risk $R(Q, P)$ is mixture-preserving in\ $Q$ with
\begin{equation}
R( Q_{\lambda}, P ) = (1-\lambda) R(Q_0, P) + \lambda R(Q_1, P)
\end{equation}
due to the linearity of expectation.

Strict concavity of entropy implies that mixing two states, represented by
$Q_0$ and $Q_1$, increases entropy. 
In physics, the concavity of entropy is the fundamental principle
underlying the second law of thermodynamics.

Furthermore, when maximising entropy $H(Q)$ or $G(Q, P)$ with regard to $Q$, with or without constraints, 
strict concavity ensures that the maximum is unique (if a maximum exists).
In physics, this ensures the existence of stable equilibria. 

\subsection{Entropy under change of variables and data processing}
\label{sec:entvarchange}

A further important property of entropy is its behaviour under data processing, transforming  $x$ into $y$. 

Changing variables from $x$ to $y$ is an invertible transformation and a special case of data processing.
For a discrete random variable,  information entropy$H(Q)$ remains invariant under such transformations. However, for a continuous random variable, entropy $H(Q)$ is not invariant under a change of variables, so in general 
\begin{equation}
H(Q_y) \neq H(Q_x)\,.
\end{equation}
In contrast, information entropy with prior measure $G(Q, P)$ is fully invariant under invertible transformations for both discrete as well as continuous random variables with
\begin{equation}
G(Q_y, P_y) = G(Q_x, P_x) \,.
\end{equation}

More generally, $G(Q, P)$ also satisfies the data-processing inequality
\begin{equation}
G(Q_y,P_y) \geq G(Q_x,P_x)\,,
\end{equation}
so that entropy with prior measure is nondecreasing under a data-processing map
(stochastic or deterministic, possibly coarsening) that produces $y$ from $x$,
with identity for lossless transformations such as an invertible change of variables.
A perhaps more familiar corrollary of this is that KL divergence is never increasing under
data processing,
\begin{equation}
\ikl(Q_y,P_y) \leq \ikl(Q_x,P_x)\,.
\end{equation}

\newpage

\section{Interpretations of entropy and the principle of maximum entropy}
\label{intentmaxent}

This section discusses various interpretations of entropy in both physics and statistics,
helping to motivate the principle of maximum entropy.

\subsection{Interpretations of entropy}

In thermodynamics, entropy determines how much of the internal energy 
of a system is available to do work. Low
entropy implies that energy is concentrated and available for work, and
high entropy means energy is spread out and not available for work \citep{Leff1996}.

Similarly, information entropy $G(Q, P)$ measures the spread or dispersal of probability mass in the distribution $Q$ with respect to a reference distribution $P$.  As a special case, standard information entropy $H(Q)$  measures the spread of probability mass relative to the uniform distribution. 
When probability mass is concentrated, entropy is low. Conversely, when the probability mass is spread out, entropy is high.  Correspondingly, entropy is often closely linked with other measures of dispersion of a distribution. For instance, for the normal
distribution $P=N(\mu, \sigma^2)$ the entropy $H(P) = (\log(2 \pi \sigma^2)+1)/2$
is a function of the variance $\sigma^2$.

The combinatorial perspective on entropy (\seccite{sec:combient}) directly supports its interpretation as a measure of dispersal, as the multiplicity $W_K$ is far greater when items are spread out over all $K$ categories than when they are concentrated in a small number of classes. 

The link of entropy with the log-loss (\seccite{sec:defent}) also allows interpretation in terms
of reward and minimum expected loss. Specifically, the entropy $G(Q, P)$ can be viewed as a reward when using $P$ to describe data from $Q$ that is maximised at zero when the correct model $P=Q$ is used, whereas the standard entropy $H(Q)$ is the minimum attainable mean log-loss $R(Q,P)$ for the correct model $P=Q$.

Finally, a popular but outdated notion of entropy is that of quantifying disorder, with
low entropy corresponding to order and high entropy to disorder. However, this interpretation is misleading as many apparently ordered systems can have high entropy, and some seemingly disordered systems can have low entropy \citep{Leff2007}.

\subsection{Principle of maximum entropy}

In statistical mechanics models, a physical system continuously explores its state space as energy is exchanged and redistributed. Equilibrium is the state the system occupies most of the time and corresponds to the typical macrostate with an overwhelmingly large number of compatible microstates, i.e. the macrostate with the largest volume. Therefore, in physics, equilibrium states maximise entropy subject to the relevant constraints  \citep{Reif1965,Mandl1988}. 

In probability models, the principle of maximum entropy is most naturally
understood via the dispersal interpretation of entropy.  A probability distribution whose mass is tightly concentrated has low entropy and carries much information. By contrast, a distribution with spread-out probability mass has high entropy and is relatively uninformative. Hence the maximum-entropy principle selects the least-informative distribution that satisfies the given constraints.

Jaynes, who played a vital role in bridging physics and statistics, is widely credited with introducing and popularising the principle of maximum entropy \citep{LevineTribus1979,Rosenkrantz1983,Jay03}.

\section{Exponential families via maximum entropy}
\label{sec-efmaxent}

Using only elementary calculations, this section shows that both the general exponential family (EF) and its subset with constant-base-function (EFC) arise naturally as maximum-entropy distributions.

\subsection{Emergence of exponential families}

In statistical mechanics, exponential families naturally appear as the canonical form of equilibrium distributions under constraints, such as a given mean energy level, and the strict concavity of entropy ensures that any such distribution is unique and stable \citep{Reif1965,Mandl1988}.

Correspondingly, in probability and statistics, exponential families (EF) arise by selecting a prior measure $B$ and canonical statistics $\bt(x)$, then seeking the maximum-entropy and thus least-informative distribution among all distributions with specified  expectation $\expect(\bt(x)) = \bmu_{\bt}$.  Again, uniqueness is guaranteed by concavity.
This is a classic result established in \citet{Jaynes1968}
and subsequently by \citet{ShoreJohnson1981} using constrained optimisation, see also Section~2.4.7 in \citet{Murphy2023}.
Maximising standard information entropy, rather than entropy with general prior measure, yields the subset of exponential families with a constant base function (EFC), see  \citet{Jaynes1957a}.

Derivations of general exponential families (EF) are uncommon in contemporary texts on exponential families and information theory, whereas derivations of 
the subset with a constant base function (EFC) are more frequent,  see \tabcite{tab-textbooks}.

\subsection{Exponential families with general base function maximise information entropy
with prior measure}

The following simple proposition shows that exponential families with general base function are indeed characterised by achieving maximum entropy among all distributions with the same
support and specified expectation constraint.   

Assume that $P= P(\bfeta)$ is a general exponential family constructed by exponential tilting of the base distribution $B=P(\bfeta_0)$ using the canonical statistics $\bt(x)$.
The expectation parameters corresponding to $\bfeta$ are $\expect_P(\bt(x)) = \bmu_{\bt}$. Let $Q$ be a distribution on the same domain as $B$ and $P$
satisfying the constraint $\expect_Q(\bt(x)) = \bmu_{\bt}$.

\begin{prop}
The information entropy $G(Q, B)$ of $Q$ with prior measure $B$ is uniquely maximised with respect to $Q$ at $Q=P$ so that the entropy $G(P, B)$ of the exponential family $P$ constructed from base $B$ and canonical statistics $\bt(x)$ with expectation $\bmu_{\bt}$ is larger than the entropy  $G(Q, B)$ unless $Q=P$.
\end{prop}

\begin{proof}

Using \eqcite{eq:expfampdfb} the log-pdmf of $P$ can be written as 
\begin{equation}
\log p(x| \bfeta) = \langle \bfeta-\bfeta_0, \bt(x) \rangle + \log b(x) - a(\bfeta) + a(\bfeta_0)
\end{equation}
where $b(x)$ is the pdmf of the base distribution $B$.
Correspondingly, with \eqcite{eq:infoent} and \eqcite{eq:meanlogloss} the entropy of $P$ is  
\begin{equation}
H(P) = -\langle \bfeta-\bfeta_0, \bmu_{\bt} \rangle + R(P, B) + a(\bfeta)-a(\bfeta_0)
\label{eq:entef}
\end{equation}
and the mean log-loss of $P$ under $Q$ is 
\begin{equation}
R(Q, P) = -\langle \bfeta-\bfeta_0, \bmu_{\bt} \rangle + R(Q, B) + a(\bfeta)-a(\bfeta_0)\,.
\label{eq:mllef}
\end{equation} 
Combining \eqcite{eq:entef} and \eqcite{eq:mllef} yields
\begin{equation}
R(Q, P) = H(P)-R(P,B)+R(Q,B) \,.
\label{eq:mllef2}
\end{equation}
The properness \ineqcite{eq:properness} states that $R(Q, P) \geq H(Q)$, with identity uniquely for $Q=P$.
Thus, from \eqcite{eq:mllef2} it follows 
\begin{equation}
H(P)-R(P,B) \geq H(Q)-R(Q,B)
\end{equation}
and hence 
\begin{equation}
G(P, B) \geq G(Q, B)
\end{equation} 
with identity uniquely for $Q=P$.
\end{proof}

Therefore, the exponential family $P$ with base $B$ achieves maximum entropy with reference measure $B$ and any distribution $Q$ with the same mean constraint will have lower entropy than $P$, unless $Q=P$.
The maximum is unique because of the properness \ineqcite{eq:properness} and the strict concavity of information entropy with prior measure (\eqcite{eq:entpconc}).

Note that information entropy with prior measure
and the general exponential family are both invariant under a change of variables
(cf. \seccite{sec:efvarchange} and \seccite{sec:entvarchange}).

\subsection{Exponential families with constant base function maximise standard information entropy}

A similar but even simpler argument  applies to the special case of exponential families with a constant base function (EFC).

Assume that $P= P(\bfeta)$ is an exponential family with constant base function
$h(x)=1$, canonical statistics $\bt(x)$, and expectation parameters  $\expect_P(\bt(x)) = \bmu_{\bt}$. Let $Q$ be a distribution on the same domain as $P$
satisfying the constraint $\expect_Q(\bt(x)) = \bmu_{\bt}$.

\begin{prop}
The information entropy $H(Q)$ of $Q$ is uniquely maximised with respect to $Q$ at $Q=P$ so that the entropy $H(P)$ of the exponential family $P$ with a constant base function and canonical statistics $\bt(x)$ with expectation $\bmu_{\bt}$ is larger than the entropy  $H(Q)$ unless $Q=P$.
\end{prop}

\begin{proof}

From \eqcite{eq:expfampdfc}
assuming $P$ has constant base function $h(x)=1$ the log-pdmf of $P$ can be written as
\begin{equation}
\log p(x| \bfeta) = \langle \bfeta, \bt(x) \rangle ) - a(\bfeta)
\end{equation}
Correspondingly, the entropy of $P$ is  
\begin{equation}
H(P) = -\langle \bfeta, \bmu_{\bt} \rangle  + a(\bfeta)
\end{equation}
and the mean log-loss of $P$ under $Q$ is 
\begin{equation}
R(Q, P) = -\langle \bfeta, \bmu_{\bt} \rangle + a(\bfeta)\,.
\end{equation}
Therefore, 
\begin{equation}
R(Q, P) = H(P) \,.
\label{eq:mllefc2}
\end{equation}
The properness \ineqcite{eq:properness} states that $R(Q, P) \geq H(Q)$, with identity uniquely for $Q=P$.
Thus, from \eqcite{eq:mllefc2} follows  
\begin{equation}
H(P) \geq H(Q)
\end{equation}
with identity uniquely for $Q=P$.

\end{proof}

Therefore, an exponential family $P$ with constant base function achieves maximum entropy $H(P)$ and any distribution $Q$ with the same mean constraint will have lower information entropy than $P$, unless $Q=P$.
The maximum is unique because of the properness \ineqcite{eq:properness} and the strict concavity of information entropy (\eqcite{eq:entconc}).

Note that standard information entropy
and exponential families with a constant base function are both not invariant under a change of variables (cf. \seccite{sec:efvarchange} and \seccite{sec:entvarchange}).

\section{Conclusion}

The justification of the functional form of exponential families is often treated as technically difficult and left to advanced texts. 
The purpose of this note is to show that the principle of maximum entropy allows for an elementary, back-of-the-envelope derivation requiring only basic knowledge in information theory.

Proposition~1 covers general exponential families (EF) and requires maximisation of information entropy with prior measure.  Proposition~2 covers the restricted case focusing on the subset of exponential families with constant base function (EFC) and uses maximisation of standard information entropy.
Similar maximum-entropy derivations have been previously discussed, in particular for EFCs.  Proposition~2 is a simplified version of Theorem 12.1.1 of \cite{CoverThomas2006}.

Correspondingly, the principle of maximum entropy provides an intuitive rationale for exponential families as the least-informative distributions given moment constraints. Thus, they offer a principled default when only expectations of canonical statistics are known and no further structure is assumed.

The close relationship between entropy and exponential families highlights the importance of teaching information entropy more broadly in statistics, since entropy-based arguments are central to many other statistical applications  \citep{KonishiKitagawa2008}.  Moreover, information entropy, and its generalisations,  underpin modern generalised learning procedures  \citep{KoblauchJewsonDamoulas2022,KhanRue2023}, with direct links to proper scoring rules \citep{Dawid2007} and information geometry \citep{Amari2016}.

\section*{Acknowledgements}

I thank the reviewers for their detailed and constructive comments.

\bibliographystyle{apalike}
\bibliography{preamble,stats,physics,strimmer}

\newcommand{\noopsort}[1]{} \newcommand{\printfirst}[2]{#1}
  \newcommand{\singleletter}[1]{#1} \newcommand{\switchargs}[2]{#2#1}
\begin{thebibliography}{}

\bibitem[Akaike, 1985]{Akaike1985}
Akaike, H. (1985).
\newblock Prediction and entropy.
\newblock In Atkinson, A. and Fienberg, S., editors, {\em A Celebration of
  Statistics}, chapter~1, pages 1--24. Springer.

\bibitem[Amari, 2016]{Amari2016}
Amari, S. (2016).
\newblock {\em Information Geometry and Its Applications}.
\newblock Springer.

\bibitem[Barankin and Maitra, 1963]{BarankinMaitra1963}
Barankin, E.~W. and Maitra, A.~P. (1963).
\newblock Generalization of the {Fisher}-{Darmois}-{Koopman}-{Pitman} theorem
  on sufficient statistics.
\newblock {\em Sankhy\={a} A}, 25:217--244.

\bibitem[Cover and Thomas, 2006]{CoverThomas2006}
Cover, T.~M. and Thomas, J.~A. (2006).
\newblock {\em Elements of Information Theory}.
\newblock John Wiley \& Sons, 2nd edition.

\bibitem[Darmois, 1935]{Darmois1935}
Darmois, G. (1935).
\newblock Sur les lois de probabilit\'{e} \`{a} estimation exhaustive.
\newblock {\em C. R. Acad. Sci. Paris}, 200:1265--1266.

\bibitem[Dawid, 2007]{Dawid2007}
Dawid, A. (2007).
\newblock The geometry of proper scoring rules.
\newblock {\em Ann. Inst. Statist. Math.}, 59:77--93.

\bibitem[Dawid and Musio, 2014]{DawidMusio2014}
Dawid, A. and Musio, M. (2014).
\newblock Theory and applications of proper scoring rules.
\newblock {\em METRON}, 72:169--183.

\bibitem[Efron, 2022]{EfronExpFam2022}
Efron, B. (2022).
\newblock {\em Exponential Families in Theory and Practise}.
\newblock Cambridge University Press.

\bibitem[Ellis, 1999]{Ellis1999}
Ellis, R.~S. (1999).
\newblock The theory of large deviations: from {Boltzmann's} 1877 calculation
  to equilibrium macrostates in 2d turbulence.
\newblock {\em Physica D}, 133:106--136.

\bibitem[Esscher, 1932]{Esscher1932}
Esscher, F. (1932).
\newblock On the probability function in the collective theory of risk.
\newblock {\em Scand. Actuar. J.}, 15:175--195.

\bibitem[Fisher, 1934]{Fisher1934}
Fisher, R.~A. (1934).
\newblock Two new properties of mathematical likelihood.
\newblock {\em Proc. R. Soc. Lond. A}, 144:285--307.

\bibitem[Golan, 2018]{Golan2018}
Golan, A. (2018).
\newblock {\em Foundations of Info-Metrics: Modeling, Inference, and Imperfect
  Information}.
\newblock Oxford University Press.

\bibitem[Jaynes, 1957]{Jaynes1957a}
Jaynes, E.~T. (1957).
\newblock Information theory and statistical mechanics.
\newblock {\em Phys. Rev.}, 106:620--630.

\bibitem[Jaynes, 1968]{Jaynes1968}
Jaynes, E.~T. (1968).
\newblock Prior probabilities.
\newblock {\em IEEE Trans. Syst. Sci. Cybern.}, 4:227--241.

\bibitem[Jaynes, 2003]{Jay03}
Jaynes, E.~T. (2003).
\newblock {\em Probability Theory: The Logic of Science}.
\newblock Cambridge University Press.
\newblock Edited by G. L. Bretthorst.

\bibitem[Khan and Rue, 2023]{KhanRue2023}
Khan, M.~E. and Rue, H. (2023).
\newblock The bayesian learning rule.
\newblock {\em J. Mach. Learn. Res.}, 24(281):1--46.

\bibitem[Knoblauch et~al., 2022]{KoblauchJewsonDamoulas2022}
Knoblauch, J., Jewson, J., and Damoulas, T. (2022).
\newblock An optimization-centric view on bayes' rule: Reviewing and
  generalizing variational inference.
\newblock {\em J. Mach. Learn. Res.}, 23(132):1--109.

\bibitem[Konishi and Kitagawa, 2008]{KonishiKitagawa2008}
Konishi, S. and Kitagawa, G. (2008).
\newblock {\em Information Criteria and Statistical Modeling}.
\newblock Springer.

\bibitem[Koopman, 1936]{Koopman1936}
Koopman, B.~O. (1936).
\newblock On distributions admitting a sufficient statistic.
\newblock {\em Trans. Amer. Math. Soc.}, 39:399--409.

\bibitem[Kullback and Leibler, 1951]{KullbackLeibler1951}
Kullback, S. and Leibler, R.~A. (1951).
\newblock On information and sufficiency.
\newblock {\em Ann. Math. Statist.}, 22:79--86.

\bibitem[Leff, 1996]{Leff1996}
Leff, H.~S. (1996).
\newblock Thermodynamic entropy: The spreading and sharing of energy.
\newblock {\em Am. J. Phys.}, 64:1261--1271.

\bibitem[Leff, 2007]{Leff2007}
Leff, H.~S. (2007).
\newblock Entropy, its language, and interpretation.
\newblock {\em Bell Syst. Tech. J.}, 77:1744--1766.

\bibitem[Levine and Tribus, 1979]{LevineTribus1979}
Levine, R.~D. and Tribus, M., editors (1979).
\newblock {\em The Maximum Entropy Formalism}.
\newblock MIT Press.

\bibitem[Little, 2019]{Little2019}
Little, M.~A. (2019).
\newblock {\em Machine Learning for Signal Processing: Data Science,
  Algorithms, and Computational Statistics}.
\newblock Oxford University Press.

\bibitem[Mandelbrot, 1962]{Mandelbrot1962}
Mandelbrot, B. (1962).
\newblock The role of sufficiency and of estimation in thermodynamics.
\newblock {\em Ann. Math. Statist}, 33:1021--1038.

\bibitem[Mandl, 1988]{Mandl1988}
Mandl, F. (1988).
\newblock {\em Statistical Physics}.
\newblock Wiley, 2nd edition.

\bibitem[McElreath, 2020]{McElreath2020}
McElreath, R. (2020).
\newblock {\em Statistical Rethinking}.
\newblock Chapman and Hall/CRC, 2nd edition.

\bibitem[Morris and Lock, 2009]{MorrisLock2009}
Morris, C.~N. and Lock, K.~F. (2009).
\newblock Unifying the named natural exponential families and their relatives.
\newblock {\em Am. Stat.}, 63:247--253.

\bibitem[Murphy, 2012]{Murphy2012}
Murphy, K.~P. (2012).
\newblock {\em Machine Learning: A Probabilistic Perspective}.
\newblock MIT Press.

\bibitem[Murphy, 2023]{Murphy2023}
Murphy, K.~P. (2023).
\newblock {\em Probabilistic Machine Learning: Advanced Topics}.
\newblock MIT Press.

\bibitem[Ovcharov, 2018]{Ovcharov2018}
Ovcharov, E.~Y. (2018).
\newblock Proper scoring rules and {Bregman} divergences.
\newblock {\em Bernoulli}, 24:53--79.

\bibitem[Pachter et~al., 2024]{PachterYangDill2024}
Pachter, J.~A., Yang, Y.-J., and Dill, K.~A. (2024).
\newblock Entropy, irreversibility and inference at the foundations of
  statistical physics.
\newblock {\em Nat. Rev. Physics}, 6:382--393.

\bibitem[Pitman, 1936]{Pitman1936}
Pitman, E. J.~G. (1936).
\newblock Sufficient statistics and intrinsic accuracy.
\newblock {\em Math. Proc. Camb. Philos. Soc.}, 32:567--579.

\bibitem[Reif, 1965]{Reif1965}
Reif, F. (1965).
\newblock {\em Fundamentals of Statistical and Thermal Physics}.
\newblock McGraw-Hill.

\bibitem[Rosenkrantz, 1983]{Rosenkrantz1983}
Rosenkrantz, R.~D., editor (1983).
\newblock {\em E. T. Jaynes: Papers on Probability, Statistics and Statistical
  Physics}.
\newblock D. Reidel Publishing Company.

\bibitem[Shore and Johnson, 1981]{ShoreJohnson1981}
Shore, J. and Johnson, R. (1981).
\newblock Properties of cross-entropy minimization.
\newblock {\em IEEE Trans. Inform. Theory}, 27:472--482.

\bibitem[Sundberg, 2019]{Sundberg2019}
Sundberg, R. (2019).
\newblock {\em Statistical Modelling by Exponential Families}.
\newblock Cambridge University Press.

\bibitem[Wainwright and Jordan, 2008]{WainwrightJordan2008}
Wainwright, M.~J. and Jordan, M.~I. (2008).
\newblock Graphical models, exponential families, and variational inference.
\newblock {\em Found. Trends Mach. Learn.}, 1:1--305.

\bibitem[Wehrl, 1978]{Wehrl1978}
Wehrl, A. (1978).
\newblock General properties of entropy.
\newblock {\em Rev. Mod. Phys.}, 50:221--260.

\end{thebibliography}

\end{document}